\documentclass[pra,superscriptaddress,twocolumn]{revtex4}
\usepackage{amsfonts,amssymb,amsmath}
\usepackage[]{graphics,graphicx,epsfig}
\usepackage{amsthm}
\usepackage{color}

\usepackage{amsmath,amsthm,amsfonts,amssymb,amscd}
\usepackage[utf8]{inputenc}
\usepackage{indentfirst}
\usepackage{graphicx}
\usepackage[T1]{fontenc}
\usepackage{mathpazo}
\usepackage{changes}

\usepackage[USenglish]{babel}
\usepackage[colorlinks=true]{hyperref}
\usepackage{hyperref}
\usepackage{bm}
\usepackage{lscape}
\newcommand{\de}[1]{\left( #1 \right)}
\newcommand{\De}[1]{\left[#1\right]}
\newcommand{\DE}[1]{\left\{#1\right\}}

\def\id{\leavevmode\hbox{\small1\kern-3.8pt\normalsize1}}

\def\identity{\leavevmode\hbox{\small1\kern-3.8pt\normalsize1}}

\newtheorem{mydef}{Definition}
\newtheorem{appendixdef}{Definition}[section]

\renewcommand{\epsilon}{\varepsilon}

\newtheorem{definition}{Definition} 

\newtheorem{lemma}[definition]{Lemma}

\newtheorem{thm}[definition]{Theorem}
\newtheorem{apptheorem}{Theorem}[section]

\newtheorem*{rep@theorem}{\rep@title}
\newcommand{\newreptheorem}[2]{%
\newenvironment{rep#1}[1]{%
 \def\rep@title{#2 \ref{##1} (restatement)}%
 \begin{rep@theorem}}%
 {\end{rep@theorem}}}
\makeatother

\newreptheorem{thm}{Theorem}
\newreptheorem{lem}{Lemma}

\def\ba#1\ea{\begin{align}#1\end{align}}
\def\ban#1\ean{\begin{align*}#1\end{align*}}

\newcommand{\ot}{\otimes}
\newcommand{\be}{\begin{equation}}
\newcommand{\ee}{\end{equation}}

\def\benum{\begin{enumerate}}
\def\eenum{\end{enumerate}}

\def\squareforqed{\hbox{\rlap{$\sqcap$}$\sqcup$}}
\def\qed{\ifmmode\squareforqed\else{\unskip\nobreak\hfil
\penalty50\hskip1em\null\nobreak\hfil\squareforqed
\parfillskip=0pt\finalhyphendemerits=0\endgraf}\fi}
\def\endenv{\ifmmode\;\else{\unskip\nobreak\hfil
\penalty50\hskip1em\null\nobreak\hfil\;
\parfillskip=0pt\finalhyphendemerits=0\endgraf}\fi}
\newenvironment{example}{\noindent \textbf{{Example~}}}{\qed}

\newcommand{\<}{\langle}
\renewcommand{\>}{\rangle}

\def\id{{\operatorname{id}}}

\def\be{\begin{equation}}
\def\ee{\end{equation}}
\def\ben{\begin{eqnarray}}
\def\een{\end{eqnarray}}
\def\ot{\otimes}

\def\bei{\begin{itemize}}
\def\eei{\end{itemize}}

\mathchardef\ordinarycolon\mathcode`\:
\mathcode`\:=\string"8000
\def\vcentcolon{\mathrel{\mathop\ordinarycolon}}
\begingroup \catcode`\:=\active
  \lowercase{\endgroup
  \let :\vcentcolon
  }

\newcommand{\nc}{\newcommand}
 \nc{\proj}[1]{|#1\rangle\!\langle #1 |} 
\nc{\avg}[1]{\langle#1\rangle}

\nc{\conv}{\operatorname{conv}}
\nc{\smfrac}[2]{\mbox{$\frac{#1}{#2}$}} \nc{\Tr}{\operatorname{Tr}}
\nc{\ox}{\otimes} \nc{\dg}{\dagger} \nc{\dn}{\downarrow}
\nc{\lmax}{\lambda_{\text{max}}}
\nc{\lmin}{\lambda_{\text{min}}}

\nc{\csupp}{{\operatorname{csupp}}}
\nc{\qsupp}{{\operatorname{qsupp}}} \nc{\var}{\operatorname{var}}
\nc{\rar}{\rightarrow} \nc{\lrar}{\longrightarrow}
\nc{\poly}{\operatorname{poly}}
\nc{\polylog}{\operatorname{polylog}} \nc{\Lip}{\operatorname{Lip}}
\nc{\Om}{\Omega}
\nc{\wt}[1]{\widetilde{#1}}

\def\>{\rangle}
\def\<{\langle}

\nc{\glneq}{{\raisebox{0.6ex}{$>$}  \hspace*{-1.8ex} \raisebox{-0.6ex}{$<$}}}
\nc{\gleq}{{\raisebox{0.6ex}{$\geq$}\hspace*{-1.8ex} \raisebox{-0.6ex}{$\leq$}}}

\nc{\vholder}[1]{\rule{0pt}{#1}}
\nc{\wh}[1]{\widehat{#1}}
\nc{\h}[1]{\widehat{#1}}

\nc{\ob}[1]{#1}

\def\beq{\begin {equation}}
\def\eeq{\end {equation}}

\def\be{\begin{equation}}
\def\ee{\end{equation}}

\nc{\eq}[1]{(\ref{eq:#1})} 
\nc{\eqs}[2]{\eq{#1} and \eq{#2}}

\nc{\eqn}[1]{Eq.~(\ref{eqn:#1})}
\nc{\eqns}[2]{Eqs.~(\ref{eqn:#1}) and (\ref{eqn:#2})}

\nc{\region}{\cS\cW}

\newenvironment{protocol*}[1]
  {
    \begin{center}
      \hrulefill\\
      \textbf{#1}
  }
  {
    \vspace{-1\baselineskip}
    \hrulefill
    \end{center}
  }

\begin{document}
\title{Generalized XOR games with $d$ outcomes and the task of non-local computation}
\author{Ravishankar \surname{Ramanathan}}
\email{ravishankar.r.10@gmail.com}
\affiliation{National Quantum Information Center of Gda\'nsk,  81-824 Sopot, Poland}
\affiliation{Institute of Theoretical Physics and Astrophysics, University of Gda\'{n}sk, 80-952 Gda\'{n}sk, Poland}
\author{Remigiusz Augusiak}
\affiliation{Center for Theoretical Physics, Polish Academy of Sciences, Aleja Lotników 32/46, 02-668 Warsaw, Poland}
\affiliation{ICFO–Institut de Ciencies Fotoniques, Mediterranean Technology Park, 08860 Castelldefels (Barcelona), Spain}
\author{Gl\'aucia Murta}
\affiliation{Departamento de Fisica, Universidade Federal de Minas Gerais, Caixa Postal 702, 30123-970, Belo Horizonte, MG, Brazil}
\begin{abstract}
Two-party \textsc{xor} games (correlation Bell inequalities with two outcomes per party) are the most studied Bell inequalities, and one of the few classes for which the optimal quantum value is known to be exactly calculable.
We study a natural generalization of the binary \textsc{xor} games  to the class of linear games with $d > 2$ outcomes, and propose an easily computable bound on the quantum value of these games. Many interesting properties such as the impossibility of a quantum strategy to win these games, and the quantum bound on the CHSH game generalized to $d$ outcomes are derived. 
We also use the proposed bound to prove a large-alphabet generalization of the principle of no quantum advantage in non-local computation, showing that quantum theory provides no advantage in the task of non-local distributed computation of a class of functions with $d$ outcomes for prime $d$, while general no-signaling boxes do. This task is one of the information-theoretic principles attempting to characterize the set of quantum correlations from amongst general no-signaling ones. 
\end{abstract}

\maketitle

\section{Introduction.}
Quantum non-local correlations are one of the most intriguing aspects of Nature, evidenced in the violation of Bell inequalities. Besides their foundational interest, these correlations have also proven to be useful in information processing tasks such as secure device-independent randomness amplification and expansion \cite{rand}, cryptographic secure key generation \cite{securekey} and reduction of communication complexity \cite{CommCompl}. 

Concerning such applications, it is typically of most interest to compute the classical and quantum value of the Bell expression, the classical value being the maximum over local realistic assignments of outcomes while the quantum value is the maximum attained using measurements on entangled quantum states. However, neither of these values is easy to calculate. Computing the classical value is done by means of an integer program and is in general a hard problem \cite{Pitowsky, Hastad}. On the other hand, it is not even known whether the quantum value is computable for all Bell inequalities, since there is a priori no restriction on the dimension of the Hilbert space for the quantum states and measurements; although in some instances it is possible to compute the value efficiently or to find a good approximation. A hierarchy of semi-definite programs from \cite{Navascues} is typically used to get (upper) bounds on the quantum value, although the quality of approximation achieved by these bounds remains unknown. The size of these programs also increases exponentially with the number of inputs and outputs in the Bell expression, so that a central problem of utmost importance in non-locality theory is to find easily computable good bounds to handle general classes of Bell inequalities.  

An important class of Bell inequalities for which the quantum value \textit{can} be computed exactly is the class known as two-party binary \textsc{xor} games or equivalently as bipartite two-outcome correlation inequalities. In a binary \textsc{xor} game, the two parties Alice and Bob receive inputs $x \in [m_A], y \in [m_B]$ (we denote $[m_A] := \{1,\dots, m_A\}$) and respond with outputs $a, b \in \{0,1\}$. The winning constraint for each pair of inputs $(x, y)$ only depends on the \textsc{xor} modulo $2$ of the parties' answers, i.e., the Bell expression in the binary \textsc{xor} game only involves probabilities $P(a \oplus_{2} b = k| x, y)$ for $k \in \{0,1\}$. The fact that these are equivalent to Bell inequalities for correlation functions with binary outcomes is seen by noting that in this case the correlators $\mathcal{E}_{x,y}$ are given by $\mathcal{E}_{x,y} = \sum_{k=0,1} (-1)^k P(a \oplus_{2} b = k | x, y)$. For these games, it was shown in \cite{Cleve, Wehner} based upon a theorem by Tsirelson \cite{Tsirelson} that the quantum value can be computed efficiently by means of a semi-definite program, although computing the classical value is known to be a hard problem even for this class of games \cite{Hastad}. Besides binary \textsc{xor} games, few general results are known regarding the maximum quantum violation of classes of Bell inequalities.

The study of correlation Bell inequalities for binary outcomes was in part driven by the fact that many of the quantum information-processing protocols were developed for qubits, for which binary outcome games appear naturally. Recently, there has been much interest in developing applications of higher-dimensional entanglement \cite{Exp-high-dim, Qudit-Toffoli, Qudit-randomness, Qudit-key-dist} for which Bell inequalities with more than two outcomes may be naturally suited. Therefore, both for fundamental reasons as well as for these applications, the study of Bell inequalities with more outcomes is crucial.

A natural extension of the binary outcome \textsc{xor} games is to the class of generalized \textsc{xor}-d games, where the outputs of the two parties are not restricted to be binary, although the winning constraint still depends upon the generalized \text{xor} (addition modulo $d$), with $d$ being the number of outcomes. 
The generalization can also be extended to the class known as \textsc{linear} games \cite{Hastad}, where the parties output answers that are elements of a finite Abelian group and the winning constraint depends upon the group operation acting on the outputs. Linear games are the paradigmatic example of non-local games with more than two outcomes, and a study of their classical and quantum values is crucial, especially in light of applications such as \cite{Jed}. 
In the context of Bell inequalities, these were first studied in \cite{Buhrman} where a large alphabet generalization of the CHSH inequality called CHSH-d was considered, which has since been investigated in \cite{Liang, Ji, Bavarian, Howard, GRRH+15}, 
An important property of the \textsc{xor}-d games concerns their relationship with communication complexity, following \cite{vanDam, Wang} it is seen that correlations (boxes) winning a non-trivial total function \textsc{xor}-d game for prime $d$ can result in a trivialization of communication complexity. 
A related information-theoretic principle called \textit{no quantum advantage in non-local computation} (no-NLC) has also been suggested in \cite{NLC}; this proposes that quantum correlations are those that do not provide any advantage over classical correlations in the task of distributed non-local computation of arbitrary binary functions, while general no-signaling correlations do. It is also of interest to investigate whether the above principle can be extended to functions of more outcomes.  


In this paper, we present a novel efficiently computable bound to the quantum value of linear games and use it to derive several interesting properties, with particular emphasis on the important case of \textsc{xor}-d games for prime $d$. We illustrate the bound with the example of the CHSH-d game for prime and prime power $d$, recovering recent results derived using alternative (more technical) methods.
As another illustration, we use the bound to show that for uniformly chosen inputs, no non-trivial total function \textsc{xor}-d game can be won with a quantum strategy and consequently that
these no-signaling boxes that trivialize communication complexity cannot be realized within quantum theory. We further prove a  large alphabet generalization of the no-NLC principle, showing that quantum theory provides no advantage in the task of non-local computation of a restricted class of functions with $d$ outcomes for prime $d$. 
For the sake of clarity of exposition, we only include sketches of proofs in the main text with details deferred to the Appendices.

\section{A bound on the quantum value of linear games.}
\label{sec:lin-bound}
Linear games are a generalization of \textsc{xor} games to an arbitrary output alphabet size and are defined as follows:
\begin{mydef}
A two-player linear game $\textsl{g}^{l} = (q, f)$ is one where two players Alice and Bob receive questions $u$, $v$ from sets $Q_A$ and $Q_B$ respectively, chosen from a probability distribution $q(u,v)$ by a referee. They reply with respective answers $a, b \in (G,+)$ where $G$ is a finite Abelian group with associated operation $+$. The game is defined by a winning constraint $a + b = f(u,v)$ for some function $f : Q_A \times Q_B \rightarrow G$.  
\end{mydef}

The most interesting linear games are arguably the \textsc{xor}-d games, denoted $\textsl{g}^{\oplus}$ which are the linear games corresponding to the cyclic group $\mathbb{Z}_d$, the integers with operation addition modulo $d$ ($\oplus_d$). 
The value of the linear game is given by the expression 
\begin{equation}
\omega_s(\textsl{g}^{l}) = \max_{\{P_{A,B|U,V}\} \in \mathcal{S}} \sum_{\substack{u \in Q_A \\ v \in Q_B}} \sum_{a,b \in G} q(u,v) V(a,b|u,v) P(a,b | u, v),
\end{equation}
where $V(a,b|u,v) = 1$ if $a + b =f(u,v)$ and $0$ otherwise and the maximum is taken over all boxes $\{P_{A,B|U,V}\}$ in the set $\mathcal{S}$ which may correspond to the set of classical $\mathcal{C}$, quantum $\mathcal{Q}$ or more general no-signaling boxes $\mathcal{NS}$. 
The maximum classical value of the game (the maximum over all deterministic assignments of $a, b$ for each respective input $u,v$ or their convex combinations) is denoted $\omega_c(\textsl{g}^{l})$, the maximum value of the game achieved by a quantum strategy (POVM measurements on a shared entangled state of arbitrary Hilbert space dimension) is denoted $\omega_q(\textsl{g}^{l})$, while the maximum value achieved by a no-signaling strategy (where neither party can signal their choice of input using the correlations) is denoted $\omega_{ns}(\textsl{g}^{l})$. These games have been studied \cite{Hastad, Khot} in the context of hardness of approximation of several important optimization problems, in attempts to identify the existence of polynomial time algorithms to approximate the optimum solution of the problem to within a constant factor. Linear games belong to the class of unique games \cite{Kempe}; in a unique game $\textsl{g}^u$ for every answer $a$ of Bob, there is a unique answer $b = \pi_{u,v}(a)$ that wins the game, where $\pi_{u,v}$ is some permutation that depends on the input pair $(u,v)$. For every game in this class, a no-signaling box exists that wins the game, so that $\omega_{ns}(\textsl{g}^l) = \omega_{ns}(\textsl{g}^u) = 1$. Such a box for the general unique game with $d$ outcomes is defined by the entries $P(a,b|u,v) = 1/d$ if $b = \pi_{u,v}(a)$ and $0$ otherwise for all input pairs $(u,v)$, this strategy clearly wins the game, and is no-signaling since the output distribution seen by each party is fully random for every input, i.e., $P(a|u) = P(b|v) = 1/d$.  

As in the case of Boolean functions \cite{BBLV,Brassard2}, the classical value $\omega_c(g^l)$ for any linear game is strictly greater than the pure random guess value $1/|G|$, this is shown in Lemma \ref{lem:cl-value}. 
\begin{lemma}
\label{lem:cl-value}
For any linear game $g^{l}$ corresponding to a function $f(u,v)$ with $u \in Q_A, v \in Q_B$ and for an arbitrary probability distribution $q(u,v)$, we have 
\begin{equation}
\label{cl-low-b}
\omega_c(g^{l}) \geq \frac{1}{|G|} \left( 1 + \frac{|G|-1}{m}\right),
\end{equation}
where $m = \min\{|Q_A|, |Q_B|\}$. 
\end{lemma}
\begin{proof}
Let $d = |G|$, Alice and Bob receive inputs $u, v$ of $\log_d |Q_A|$ and $\log_d |Q_B|$ dits respectively. Suppose w.l.o.g that $|Q_A| \leq |Q_B|$ ($m = |Q_A|$), and let the two parties share a uniformly distributed random variable $w$ of $\log_d |Q_A|$ dits. The following classical strategy achieves the lower bound in Eq.(\ref{cl-low-b}). Bob outputs $b = f(w,v)$, while Alice checks if $u = w$ and if so outputs $a = e$; if not she outputs a uniformly distributed $a \in G$. In the case when $u = w$ which happens with probability $\frac{1}{m}$, $a + b = e + f(w,v) = f(u,v)$ and the strategy succeeds. When $u \neq w$, we have that $a + f(w,v)$ is uniformly random since $a$ is uniform, and the strategy succeeds with probability $\frac{1}{d}$. The value achieved by this strategy is therefore $\frac{1}{m} + \left(1 - \frac{1}{m}\right) \frac{1}{d}$.
\end{proof}
Computing the quantum value of the linear game is an onerous task, for which efficiently computable bounds are hard to find. 
We now present a bound on the quantum value of a linear game  in Theorem \ref{thm2} by using the norms of a set of \textit{game matrices} defined using the characters of the associated group. The detailed derivation of the bound is shown in the proof of this theorem presented in the Appendix \ref{sec:lin-bound-app}, and the utility and possible tightness of the bound (in scenarios such as the CHSH-d game that is applicable to tasks such as relativistic bit commitment \cite{Jed}) is considered in this section.
\begin{thm}\label{thm2}
\label{norm-bound}
The quantum value of a linear game $\textsl{g}^l$ with input sets $ Q_A, Q_B$  can be bounded as 
\begin{eqnarray}
\label{xor-d-bound-2}
\omega_{q}(\textsl{g}^{l}) \leq \frac{1}{|G|} \left[ 1 + \sqrt{|Q_A| |Q_B|} \sum_{x \in G\setminus \{e\}} \Vert \Phi_{x} \Vert \right],
\end{eqnarray}
where 
$\Phi_{x} = \sum_{(u,v) \in Q_A \times Q_B} q(u,v) \chi_{x}(f(u,v)) | u \rangle \langle v|$ are the game matrices, $\chi_{x}$ are the characters of the group $G$ and $\Vert \cdot \Vert$ denotes the spectral norm. In particular, for an \textsc{xor}-d game with $m_A$ and $m_B$ inputs for the two parties, the quantum value can be bounded as
\begin{eqnarray}
\label{eq:xor-d-bound-3}
\omega_{q}(\textsl{g}^{\oplus}) \leq \frac{1}{d} \left[ 1 + \sqrt{m_A m_B} \sum_{k= 1}^{d-1} \Vert \Phi_{k} \Vert \right],
\end{eqnarray}
with $\Phi_k = \sum_{\substack{u \in [m_A] \\ v \in [m_B]}} q(u,v) \zeta^{k f(u,v)} |u \rangle \!\langle v|$ and $\zeta = \exp{(2 \pi I/d)}$.
\end{thm}
\begin{proof}
We sketch the proof of the bound using the Fourier transform for the \textsc{xor}-d games here, the generalization to linear games uses the analogous Fourier transform on finite Abelian groups \cite{Terras} and is deferred to the Appendix \ref{sec:lin-bound-app}. For a quantum strategy given by projective measurements $\{\Pi_{u}^{a} \}, \{\Sigma_{v}^{b} \}$ on a pure state $| \Psi \rangle \in \mathbb{C}^{D \times D}$, we introduce the generalized correlators $\langle A_{u}^{x} \ot B_{v}^{y} \rangle$ 
for unitary operators defined as
\begin{equation}
A_{u}^{x} = \sum_{a \in G} \zeta^{-ax} \Pi_{u}^{a} \; \;  \text{and} \; \; B_{v}^{y} = \sum_{b \in G} \zeta^{-by} \Sigma_{v}^{b}.
\end{equation}
The probabilities $P(a,b|u,v)$ that enter the game expression are calculated from the inverse transform to be
\begin{eqnarray}
\label{eq:game-prob}
P(a \oplus_d b = f(u,v) | u,v) = \frac{1}{d} \sum_{k =0}^{d-1} \zeta^{k f(u,v)} \langle A_{u}^{k} \ot B_{v}^{k} \rangle. 
\end{eqnarray}


Now, with vectors $|\alpha_{k} \rangle, |\beta_{k} \rangle$ and the \textsc{xor}-d game matrices $\Phi_{k}$ defined as
\begin{eqnarray}
&&|\alpha_{k} \rangle = \sum_{u \in Q_A} \left((A_{u}^{k})^{\dagger} \otimes \identity \right) |\Psi \rangle \otimes |u \rangle \; \; \text{,} \; \; \nonumber \\
&&|\beta_{k} \rangle = \sum_{v \in Q_B} \left(\identity \otimes B_{v}^{k} \right) |\Psi \rangle \otimes |v \rangle, \nonumber \\
&&\Phi_{k} = \sum_{(u,v) \in Q_A \times Q_B} q(u,v) \zeta^{k f(u,v)} | u \rangle \langle v|,
\end{eqnarray}
the game expression $\sum_{(u,v) \in Q_A \times Q_B} q(u,v) P(a \oplus_d b = f(u,v)|u,v)$ can be rewritten using Eq.(\ref{eq:game-prob})  as
$(1/d) \sum_{k=0}^{d-1} \langle \alpha_k | \identity \otimes \Phi_{k} | \beta_{k} \rangle$ and the norm bound in Eq.(\ref{eq:xor-d-bound-3}) follows. 
\end{proof}
It should be noted that as shown in \cite{Kempe}, the quantum value of a linear game can be efficiently approximated, to be precise for any linear game $\textsl{g}^{l}$ with $\omega_q(\textsl{g}^{l}) = 1 - \delta$, there exists an efficient algorithm to approximate this value using a semi-definite program and a rounding procedure that gives an entangled strategy achieving $\omega_q^{\text{app}}(\textsl{g}^{l}) = 1 - 4 \delta'$, where $\delta/4 \leq \delta' \leq \delta$. While this is highly significant and useful for proving results such as a parallel repetition theorem for the quantum value of such games \cite{Kempe}, it would appear to be good for approximating the quantum value when the latter is close to unity, which is not the case for simple examples like the CHSH-d game. For uniform probability inputs $q(u,v) = 1/|Q_A| |Q_B|$ or when the input distribution possesses certain symmetries, as we shall see, the simple linear algebraic bound above supplements this result and proves to be very useful to derive other interesting properties of these games. 

We first illustrate the applicability and possible tightness of the bound by considering the flagship scenario of the CHSH-d game which generalizes the well-known CHSH game to a higher dimensional output. In this game, Alice and Bob are asked questions $u, v$ chosen uniformly at random from a finite field $\mathbb{F}_d$ of size $d$ so that $q(u, v) = 1/d^2$, where $d$ is a prime, or a prime power. They return answers $a, b \in \mathbb{F}_d$ with an aim to satisfy $a \oplus b = u \cdot v$ where the arithmetic operations are from the finite field. In \cite{Bavarian}, an intensive study of this game was performed, with two significant results obtained on the asymptotic classical and quantum values of the game. We now apply Theorem \ref{thm2} to re-derive in a simple manner the upper bound for the quantum value of CHSH-d. Comparison with the numerical results of \cite{Ji, Liang} indicates that the  bound in the following example of the CHSH-d game may not be tight in general, also note that the optimum value of the game for Pauli measurements was recently derived in \cite{Howard}. 

\begin{example}[see also \cite{Bavarian}]
\textit{The quantum value of the CHSH-d game for prime and prime power $d$, i.e., $d = p^r$ where $p$ is prime and $r \geq 1$ is an integer, can be bounded as
\begin{equation}
\label{eq:Bavarian-bound}
\omega_q(CHSH-d) \leq \frac{1}{d} + \frac{d-1}{d \sqrt{d}}.
\end{equation} }
\end{example}
\begin{proof}
%
Let us consider the CHSH-d game with associated function $f(u, v) = u \cdot v$. The entries of the game matrix $\Phi_k$ for prime $d$ are by definition $\Phi_k(u,v) = q(u,v)\zeta^{k (u \cdot v)}$ where $\zeta = \exp{\frac{2 \pi I}{d}}$ and $u, v \in \{0, \dots, d-1\}$, and we consider uniform probability inputs $q(u,v) = 1/d^2$. It is readily seen that for prime $d$, the game matrices $\Phi_k$ for $k \in \{1, \dots, d-1\}$ are equal to each other up to a permutation of rows (or columns). Moreover, a direct calculation using $\sum_{j=0}^{d-1} \zeta^j = 0$ yields that $\Phi_k^{\dagger} \Phi_k =  \identity/d^3$, so that $\Vert \Phi_k \Vert = 1/d \sqrt{d}, \; \; \forall k \in [d-1]$. Substitution into Eq.(\ref{eq:xor-d-bound-3}) with $m_A = m_B = d$ yields the bound in Eq.(\ref{eq:Bavarian-bound}) for prime $d$.  

Strictly analogous results are obtained for prime power $d = p^{r}$, where $p$ is prime and $r > 1$ is an integer. Note that here the operation $u \cdot v$ in the CHSH-d game is not defined as multiplication modulo $d$, but as multiplication in the finite field $\mathbb{F}_d$, see \cite{fields, Bavarian}. The non-zero elements of $\mathbb{F}_d$ under this multiplication operation form a cyclic group of size $d-1$, and we have $a^d = a, \; \; \forall a \in \mathbb{F}_d$. Here again, 
the game matrices $\Phi_k$ for $k \in [d-1]$ are equal to each other up to a permutation of rows (or columns). By explicit calculation, using the following properties of the characters: $\chi_k(a+b) = \chi_k(a) \chi_k(b)$ for any $a, b \in \mathbb{F}_d$; $\chi_k(a) = 1 \Longleftrightarrow a = 0$  and $\sum_{a \in \mathbb{F}_d} \chi_k(a \cdot b) = 0$ for $b \neq 0$ we obtain that $\Phi_k^{\dagger} \Phi_k = \frac{1}{d^3} \identity$ for all $k$. Substituting $\Vert \Phi_k \Vert = \frac{1}{d \sqrt{d}}, \; \; \forall k \in [d-1]$ into Eq.(\ref{eq:xor-d-bound-3}) with $|Q_A| = |Q_B| = d$ yields the bound.  
\end{proof}
Given the quantum bound, a natural question is whether there are linear games where the quantum value $\omega_q(g^{l})$ equals one, i.e., can there be quantum strategies that win a linear game? The interest in the question also stems from the domain of communication complexity. Following the results of \cite{vanDam, Wang}, any non-trivial total function \textsc{xor}-d game for prime $d$ and $n$ dits as input $\textbf{u} = (u_1, \dots, u_n), \textbf{v} = (v_1, \dots, v_n)$ is won by a no-signaling box that can result in a trivialization of communication complexity. To elaborate, it was shown that any no-signaling box that wins a non-trivial total function \textsc{xor}-d game for prime $d$ must contain as a sub-box, one of the functional boxes of the form $P(a \oplus_d b = f(u,v) | u,v) = 1/d$ for $a, b, u, v \in \{0, \dots, d-1\}$; having $d^n$ copies of the box and addressing this sub-box in each, Alice and Bob can compute any function of $d$ outputs with a single dit of communication, resulting in a trivialization of communication complexity. 

We now apply the bound to exclude these boxes that result in a trivialization of communication complexity from the set of quantum boxes. In particular, the following Lemma \ref{comm-comp-lem} shows that no non-trivial game for a total function $f(u,v)$ (a total function is one which is defined for all input pairs $(u,v)$) within the class of \textsc{xor}-d games $\textsl{g}^{\oplus}$ with uniformly chosen inputs can be won by a quantum strategy, meaning that there is no pseudo-telepathy game \cite{Brassard} within this class. 
\begin{lemma}
\label{comm-comp-lem}
For \textsc{xor}-d games $\textsl{g}^{\oplus}$ corresponding to total functions with $m$ questions per player, when the input distribution is uniform $q(u, v) = 1/m^2$, $\omega_q(\textsl{g}^{\oplus}) = 1$ iff $\omega_c(\textsl{g}^{\oplus}) = 1$, i.e., when rank$(\Phi_1) = 1$.
\end{lemma}
\begin{proof}

The constraint that the input distributions of questions to the players are uniform, $q(u,v) = 1/m^2$ for all $u, v$, is equivalent to $\Vert \Phi_k \Vert \leq 1/m$ since both the maximum (absolute value) column sum and row sum matrix norms are equal to $1/m$. Now $\omega_q(\textsl{g}^{\oplus}) = 1$ requires from the bound in Eq.(\ref{eq:xor-d-bound-3}) that $\Vert \Phi_k \Vert = 1/m$ for all $k \in \{1, \dots, d-1\}$. Consider the matrix ${\Phi_1}^{\dagger} \Phi_1$ which has entries $({\Phi_1}^{\dagger} \Phi_1)_{u,v} = \sum_{w=1}^{m} q(w,u) q(w,v) \zeta^{-f(w,u) + f(w,v)}$, where $\zeta = \exp{(2 \pi I/d)}$ is the $d$-th root of unity. Let $\{ \lambda_j \}$ be the maximum eigenvector corresponding to eigenvalue $1/m^2$ of ${\Phi_1}^{\dagger} \Phi_1$, with complex entries $\lambda_j = \vert \lambda_j \vert \zeta^{{\theta}_j}$. Let the entries of the eigenvector be ordered by absolute value, $\vert \lambda_1 \vert \geq \dots \geq \vert \lambda_m \vert$ and consider the eigenvalue equation corresponding to $\lambda_1$, we have 
\begin{equation}
\sum_{v, w = 1}^{m} |\lambda_v| \zeta^{-f(w,1) + f(w,v) + \theta_v} = m^2 |\lambda_1| \zeta^{\theta_1}.
\end{equation}
Clearly the above equation can only be satisfied when $\vert \lambda_j \vert = \vert \lambda_{j'} \vert \;\; \forall j, j'$ and when the phases add, i.e., when $f(w,v) - f(w,1) + \theta_v = f(w',v') - f(w',1) + \theta_{v'}\; \; \forall v,w,v',w'$, in particular choosing $w = w'$ here, we get $f(w,v) - f(w,v') = \theta_{v'} - \theta_{v} \; \forall w,v,v'$. With all $|\lambda_{j}|$ equal, the rest of the eigenvalue equations (for $u \neq 1$) lead to similar consistent constraint equations. 
We deduce that $\omega_q(\textsl{g}^{\oplus}) = 1$ only when the columns of the game matrix $\Phi_1$ are proportional to each other, the proportionality factor between columns $k, l$ being $\zeta^{f(u,k) - f(u,l)} = \zeta^{\theta_l - \theta_k}$. In this case (with $\text{rank}(\Phi_1) = 1$), a classical winning strategy which always exists for the first column of the game matrix $\Phi_1$ can be straightforwardly extended to a classical winning strategy for the entire game, meaning $\omega_c(\textsl{g}^{\oplus}) = 1$ also.

\end{proof}
It was recently shown that all the extremal points of the no-signaling polytope for any number of inputs and outputs cannot be realized within quantum theory \cite{our}. It remains an open question whether \textit{all} such vertices lead to a trivialization of communication complexity (at least in a probabilistic setting), if so this would be a compelling reason for their exclusion from correlations that can be realized in nature. Also, note that while the exclusion of the boxes trivializing communication complexity from the quantum set is not surprising, we include it here as an illustration of the applicability of the bound. Indeed in subsequent work \cite{GRRH+15}, the techniques used in this paper have also been applied to exclude boxes that win games corresponding to partial functions $f(u,v)$ from the quantum set, this further illustrates the utility of the technique since these latter boxes do not trivialize communication complexity and therefore can't be excluded on that basis.


\section{Linear games with no quantum advantage: the task of non-local computation.}
\label{sec:nlc}
Even though the quantum non-local correlations cannot be used to transmit information, they enable the performance of several tasks impossible in the classical world, such as the expansion and amplification of intrinsic randomness, device-independent secure key generation, etc. An unexpected limitation of quantum correlations however is the fact that they do not provide any advantage over classical correlations in the performance of a fundamental information-theoretic task, namely the non-local distributed computation of Boolean functions \cite{NLC}, even though certain super-quantum no-signaling correlations do. 

Consider a Boolean function $f(z_1, \dots, z_n)$ from $n$ bits to $1$ bit. A non-local (distributed) computation of the function is defined as follows. Two parties, Alice and Bob, are given inputs $(x_1, \dots, x_n)$ and $(y_1, \dots, y_n)$ obeying $x_i \oplus_2 y_i  = z_i$, each bit $x_i, y_i$ being $0$ or $1$ with equal probability. This ensures that neither party has access to any input $z_i$ on their own. To perform the non-local computation, Alice and Bob must output bits $a$ and $b$ respectively such that $a \oplus_2 b = f(x_1 \oplus_2 y_1, \dots, x_n \oplus_2 y_n)$. Their goal is thus to maximize the probability of success in this task for some given input distribution $p(z_1, \dots z_n) = p(x_1 \oplus_2 y_1, \dots, x_n \oplus_2 y_n)$. In \cite{NLC}, it was shown that surprisingly for \textit{any} input distribution $p(z_1, \dots, z_n)$, Alice and Bob sharing quantum resources cannot do any better than classical resources (both give rise to only a linear approximation of the computation), while they could successfully perform the task if the resources they shared were limited by the no-signaling principle alone. This no-advantage in non-local computation (NANLC) was so striking that it was postulated as an information-theoretic principle that picks out quantum theory from among general no-signaling theories, in relation to the correlations that the theory gives rise to \cite{NLC}. 

The above consideration of functions with a single-bit output is important since these encapsulate all decision problems, a natural class of problems used to define computational complexity classes. In the program of characterizing quantum correlations however, we must consider functions with multi-bit outputs as well as functions with higher input and output alphabets. We now use the bound (\ref{eq:xor-d-bound-3}) to construct a generalized non-local computation task for functions with higher input-output alphabet. 
Consider the following generalization of the non-local computation task to \textsc{xor}-d games, namely the computation of the function $g(z_1, \dots, z_n)$ with $z_i \in \{0, \dots, d-1\}$ where $d$ is a prime. In these games which we label $NLC_d$, Alice and Bob receive $n$ dits $\textbf{x}_n = (x_1, \dots, x_n)$ and $\textbf{y}_n = (y_1, \dots, y_n)$ which obey $x_i \oplus_d y_i = z_i$. Their task is to output dits $a, b$ respectively such that
\begin{equation}
\label{eq:func-NLC}
a \oplus_d b = g(\textbf{x}_{n-1} \oplus_d \textbf{y}_{n-1}) \cdot (x_n\oplus_d y_{n}),
\end{equation}
where $\textbf{x}_{n-1} \oplus_d \textbf{y}_{n-1}$ is the dit-wise \textsc{xor} of the $n-1$ dits, i.e., $\{x_1 \oplus_d y_1, \dots, x_{n-1} \oplus_d y_{n-1}\}$ and $g$ is an arbitrary function from $n-1$ dits to $1$ dit. The inputs are chosen according to
\begin{align}\label{probdistr}
\frac{1}{d^{n+1}} p(\textbf{x}_{n-1} \oplus_d \textbf{y}_{n-1})
\end{align}
for $p(\textbf{x}_{n-1} \oplus_d \textbf{y}_{n-1})$ being an arbitrary probability distribution. As mentioned previously, all unique games including the \textsc{xor}-d games have no-signaling value of unity, so that in general $(1=) \omega_{ns}(NLC_d) > \omega_q(NLC_d)$.
We now present in Theorem \ref{thm-nlc} the result that the games $NLC_d$ defined above exhibit no quantum advantage, the detailed proof of this theorem is presented in Appendix \ref{sec:app-nlc}.  

\begin{thm}
\label{thm-nlc}
The games $NLC_d$ for arbitrary prime $d$ and for input distribution satisfying \eqref{probdistr} have no quantum advantage, i.e., $\omega_c(NLC_d) = \omega_q(NLC_d)$.
\end{thm}
\textit{Sketch of proof.}
Consider the games $NLC_d$ for prime $d$ and arbitrary number $n$ of input dits for each party. Denote the total number of inputs for each party by $m=d^n$, and the corresponding game matrices by ${\Phi^{(n)}_k}$. The $NLC_d$ games are composed of ``building-block games"
$G(t):=\DE{a \oplus_d b= t \cdot (x \oplus_d y)}$,
with $t \in \{0, \dots, d-1\}$.

Denote the Fourier vectors as $|f_j\rangle$, i.e., 
$|f_{j} \rangle = \left(1, \zeta^j, \dots, \zeta^{(d-1)j}\right)^T$, where as usual $\zeta = \exp{\frac{2 \pi I}{d}}$. 
We find that ${\Phi^{(n)}_k}^{\dagger} \Phi^{(n)}_k$ are block-circulant matrices and are hence diagonal in the basis formed by the tensor products of the Fourier vectors $\{|f_{i_1}\rangle \otimes \dots |f_{i_{n}} \rangle\}$ with $i_1, \dots, i_n \in \{0, \dots, d-1\}$. Explicit calculation of the maximum eigenvector yields that $\| \Phi^{(n)}_k \| = d \Lambda$ for $\Lambda := \max_{i_n \in \{0, \dots, d-1\}} \lambda(i_n)$ with $\lambda(i_n)$ being the number of times the 
game $G{(d-1 \cdot i_n)}$ appears in the first row of $\Phi^{(n)}_k$. Let $\mu \in \{0, \dots, d-1\}$ denote the value of $i_n$ for which the maximum of $\lambda(i_n)$ is achieved.  

For prime $d$, we obtain the following bound
on the quantum value in the uniform case
\begin{equation}
\label{uni-q-bound}
\omega_q(NLC_d)  \leq  \frac{1}{d}\left(1 + \frac{(d-1) \Lambda}{d^{n-1}} \right).
\end{equation}
The explicit classical strategy where Alice outputs deterministically $a = \mu x_n$ independent of her input $\textbf{x}_{n-1}$ and Bob outputs $b = \mu y_n$ 
independently of his input $\textbf{y}_{n-1}$ recovers this bound. 
\qed

Let us state some open questions in this line of research. Note that the slight restriction in Eq. (\ref{eq:func-NLC}) (a fixed dependence on $x_n \oplus_d y_n$), means that the games do not cover the entire class of functions considered in \cite{NLC}, it remains open whether there is no quantum advantage for the remaining functions in this class as well. It is also of interest to identify other tasks beyond NLC where quantum correlations do not provide an advantage over classical ones, and the bound should be useful to characterize these. Also, we remark that the original NANLC principle (and most of the other principles proposed so far) is known to not pick out exactly the set of quantum correlations since there exists a set of the so-called almost quantum correlations \cite{Navascues} that also satisfies the principle. The generalized NANLC principle subsumes the original NANLC principle, since the latter corresponds to the special case $d=2$. While we expect it to be, it remains to be checked whether the generalized NANLC principle proposed here is also satisfied by the almost quantum set. Finally, it is also of interest to find whether any of the inequalities corresponding to these games define facets of the classical polytope (a facet of a polytope is a face with dimension one less than that of the polytope). Games with this property (having $\omega_c = \omega_q$ and defining facets of the classical polytope) define non-trivial boundaries of the quantum set and it has been posed as an open question in \cite{GYNI, UPB2} whether such games exist for two-party Bell scenarios. 

\section{Conclusions.} In this paper, we have presented an easily computable bound on the quantum value of linear games, with particular emphasis on \textsc{xor}-d games for prime $d$. We have illustrated this bound by using to rule out from the quantum set a class of no-signaling boxes that result in a trivialization of communication complexity. To do this, we have shown that no uniform input total function \textsc{xor}-d game can be a pseudo-telepathy game. We have also shown how the recently discovered bound on the CHSH-d game in \cite{Bavarian} can be derived in a simple manner for prime and prime power $d$, in this context it is interesting to note that these games have recently found application in relativistic bit commitment \cite{Jed}. Finally, we have extended the NANLC principle to general prime dimensional output, showing that quantum theory provides no advantage over classical theories in the distributed non-local computation of a class of functions with prime dimensional output. 

In the future, it would be interesting to extend the proposed bound on the quantum value to classes of Bell inequalities beyond linear games, especially to the more general unique games. Further applications of the bound such as in the device-independent detection of genuine multipartite entanglement \cite{BGLP, GM} for arbitrary Hilbert space dimensions, in multi-party communication complexity, as well as in the identification of information processing tasks with no quantum advantage \cite{NLC}, are of immediate interest. 

\textit{ Acknowledgements.} We thank P. Horodecki and M. Horodecki for useful discussions, as well as Matej Pivoluska and J\c edrzej Kaniewski for useful comments on an earlier version of this manuscript. R.R. is supported by the ERC AdG grant QOLAPS and the Foundation for Polish Science TEAM project co-financed by the EU European Regional Development Fund. R. A. acknowledges support from the ERC AdG grant OSYRIS, the EU project SIQS, the Spanish project FOQUS and the John Templeton Foundation. G.M. acknowledges support from the Polish Ministry of Science and Higher Education Grant no. IdP2011 000361 and the Brazilian agency Fapemig (Fundação de Amparo à Pesquisa do estado de Minas Gerais).




%
%

\appendix
\onecolumngrid
 
\twocolumngrid

\vspace{0.2cm}
\section{Bounding the quantum value of linear games.}
\label{sec:lin-bound-app}
In what follows, we will use the notion of the characters of a finite Abelian group, defined in a standard manner as follows.
\begin{appendixdef}
Let $G$ be a finite Abelian group with $|G|$ elements, with operation $+$ and identity element $e$. A character of $G$ denoted $\chi$ is a homomorphism from $G$ to the multiplicative group of complex roots of unity:
\begin{equation}
\chi(a+b) = \chi(a) \chi(b) \; \; (a, b \in G)
\end{equation} 
The characters of $G$ form a finite group denoted $\hat{G}$ under elementwise multiplication. The identity element of $\hat{G}$ is denoted $\chi_e$ and satisfies $\chi_e(g) = 1$ for all $g \in G$.  
\end{appendixdef}
A useful property of the characters is that for any $\chi_e \neq \chi \in \hat{G}$, we have $\sum_{g \in G} \chi(g) = 0$ and that for any $e \neq g \in G$, we have $\sum_{\chi \in \hat{G}} \chi(g) = 0$. Note that the dual group $\hat{G}$ and $G$ are in fact isomorphic to each other. For each $x \in G$, let us denote by $\chi_{x}$ the image of $x$ under a fixed isomorphism of $G$ with $\hat{G}$. 

\begin{apptheorem}
\label{appthm:bound}
The quantum value of a linear game $\textsl{g}^l$ with input sets $ Q_A, Q_B$  can be bounded as 
\begin{eqnarray}
\label{xor-d-bound-2}
\omega_{q}(\textsl{g}^{l}) \leq \frac{1}{|G|} \left[ 1 + \sqrt{|Q_A| |Q_B|} \sum_{x \in G\setminus \{e\}} \Vert \Phi_{x} \Vert \right],
\end{eqnarray}
where 
$\Phi_{x} = \sum_{(u,v) \in Q_A \times Q_B} q(u,v) \chi_{x}(f(u,v)) | u \rangle \langle v|$ are the game matrices, $\chi_{x}$ are the characters of the group $G$ and $\Vert \cdot \Vert$ denotes the spectral norm. In particular, for an \textsc{xor}-d game with $m_A$ and $m_B$ inputs for the two parties, the quantum value can be bounded as
\begin{eqnarray}
\label{xor-d-bound-3}
\omega_{q}(\textsl{g}^{\oplus}) \leq \frac{1}{d} \left[ 1 + \sqrt{m_A m_B} \sum_{k= 1}^{d-1} \Vert \Phi_{k} \Vert \right],
\end{eqnarray}
with $\Phi_k = \sum_{\substack{u \in [m_A] \\ v \in [m_B]}} q(u,v) \zeta^{k f(u,v)} |u \rangle \langle v|$ and $\zeta = \exp{(2 \pi I/d)}$.
\end{apptheorem}

\begin{proof}
To derive a bound on the quantum value of a linear game $\omega_q(\textsl{g}^{l})$, we make use of the generalized Fourier transform on finite Abelian groups \cite{Terras}. Let us first note that by the fundamental theorem of finite Abelian groups, any finite Abelian group $G$ can be seen a direct product of cyclic groups as
$G \cong \mathbb{Z}_{n_1} \times \mathbb{Z}_{n_2} \times \dots \times \mathbb{Z}_{n_k}$
for some integers $n_1, \dots, n_k$, where $\times$ denotes the direct product and $\mathbb{Z}_n$ denotes the cyclic group of order $n$. Every element $x \in G$ can thus be seen as a $k$-tuple $(x_1, \dots, x_k)$ with $x_i \in \mathbb{Z}_{n_i}$. Denoting by $\chi_a$ the characters of the Abelian group $G$, we see that these can be written as
$\chi_{a}(x) = \prod_{j=1}^{k} \zeta_j^{a_j x_j}$,
where $\zeta_j = \exp{\frac{2 \pi i}{n_j}}$
is the $n_j$-th root of unity, and $a_j \in \mathbb{Z}_{n_j}$ for $j \in [k]$. The above relation gives a total of $\prod_{j=1}^{k} n_j = |G|$ (orthogonal) characters and consequently accounts for all the characters of $G$. Note that $\bar{\chi}_{a}(x) = \chi_{a}(-x)$, where $\bar{\chi}$ denotes the conjugate character, and $\chi_a(x) = \chi_{x}(a)$. We now introduce the generalized correlators $\langle A_{u}^{x} \ot B_{v}^{y} \rangle$ via the Fourier transform of probabilities $P(a,b|u,v)$ on the group, defined as 
\begin{equation}
\langle A_{u}^{x} \ot B_{v}^{y} \rangle = \sum_{a,b \in G} \bar{\chi}_{x}(a) \bar{\chi}_y(b) P(a,b|u,v).
\end{equation}  
The probabilities are then given by the inversion formula
\begin{equation}
P(a,b|u,v) = \frac{1}{|G|^2} \sum_{x,y \in G} \chi_{a}(x) \chi_{b}(y) \langle A_{u}^{x} \ot B_{v}^{y} \rangle.
\end{equation}
The marginals $\langle A_{u}^{x} \rangle$ are given by
\begin{eqnarray}
\langle A_{u}^{x} \rangle = \langle A_{u}^{x} \ot B_{v}^{e} \rangle &=& \sum_{a,b \in G} \bar{\chi}_{x}(a) \bar{\chi}_{e}(b) P(a,b|u,v) \nonumber \\  &=& \sum_{a \in G} \chi_{x}(-a) P(a|u),
\end{eqnarray}
where $e$ denotes the identity element of the group with $\chi_{e}$ being the trivial character ($\chi_{e}(b) = 1 \; \forall b \in G$) and we have used the no-signaling condition $\sum_{b \in G}P(a,b | u,v) = P(a|u)$; an analogous expression holds for $\langle B_{v}^{y} \rangle = \sum_{b \in G} \chi_{y}(-b) P(b|v)$. The normalization constraint is written as $\langle A_{u}^{e} \ot B_{v}^{e} \rangle = 1 \; \forall (u,v) \in Q_A \times Q_B$.
The probabilities $P(a,b|u,v)$ that enter the game expression can therefore be evaluated as
\begin{eqnarray}
&&P(a + b = f(u,v) | u,v) = \nonumber \\ &&\sum_{\substack{a,b \in G:\\ a + b = f(u,v)}} \frac{1}{|G|^2} \sum_{x,y \in G} \chi_{a}(x) \chi_{b}(y) \langle A_{u}^{x} \ot B_{v}^{y} \rangle. 
\end{eqnarray}
Using the orthogonality of the characters 
$\sum_{x \in G} \chi_{a}(x) \bar{\chi}_{b}(x) = |G| \delta_{a,b},$
where $\delta_{a,b}$ denotes the Kronecker delta, and the property of the characters that $\chi_{x}(a + b) = \chi_{x}(a) \chi_{x}(b)$ we get that
\begin{eqnarray}
\label{game-prob}
&&P(a + b = f(u,v) | u,v) = \nonumber \\ &&\sum_{a \in G} \frac{1}{|G|^2} \sum_{x,y \in G} \chi_{a}(x) \chi_{f(u,v) + a^{-1}}(y) \langle A_{u}^{x}  \ot B_{v}^{y} \rangle =\nonumber \\
&&\frac{1}{|G|} \sum_{x \in G} \chi_{f(u,v)}(x) \langle A_{u}^{x} \ot B_{v}^{x} \rangle. 
\end{eqnarray}
Now, since we do not restrict the dimension of the shared entangled states, the probabilities $P(a,b|u,v)$ are given by projective measurements $\{\Pi_{u}^{a} \}, \{\Sigma_{v}^{b} \}$ on a pure state $| \Psi \rangle \in \mathbb{C}^{D \times D}$ as
$P(a,b|u,v) = \langle \Psi | \Pi_{u}^{a} \otimes \Sigma_{v}^{b} | \Psi \rangle$
the correlators can be written as the expectation value of observables $A_{u}^{x}, B_{v}^{y}$ as
$\langle A_{u}^{x} \ot B_{v}^{y} \rangle = \langle \Psi | A_{u}^{x} \otimes B_{v}^{y} | \Psi \rangle$
with observables defined by
\begin{equation}
A_{u}^{x} = \sum_{a \in G} \bar{\chi}_{x}(a) \Pi_{u}^{a} \; \;  \text{and} \; \; B_{v}^{y} = \sum_{b \in G} \bar{\chi}_{y}(b) \Sigma_{v}^{b}.
\end{equation}
The game expression $\sum_{(u,v) \in Q_A \times Q_B} q(u,v) P(a + b = f(u,v)|u,v)$ can therefore be rewritten using Eq.(\ref{game-prob}) and the above observables as
$(1/|G|) \sum_{x \in G} \langle \alpha_x | \identity \otimes \Phi_{x} | \beta_{x} \rangle$
with vectors $|\alpha_{x} \rangle, |\beta_{y} \rangle$ and the linear game matrices $\Phi_{x}$ defined as
\begin{eqnarray}
&&|\alpha_{x} \rangle = \sum_{u \in Q_A} \left((A_{u}^{x})^{\dagger} \otimes \identity \right) |\Psi \rangle \otimes |u \rangle \; \; \text{,} \; \; \nonumber \\
&&|\beta_{y} \rangle = \sum_{v \in Q_B} \left(\identity \otimes B_{v}^{y} \right) |\Psi \rangle \otimes |v \rangle, \nonumber \\
&&\Phi_{x} = \sum_{(u,v) \in Q_A \times Q_B} q(u,v) \chi_{x}(f(u,v)) | u \rangle \langle v|.
\end{eqnarray}
The normalization of the input probability distribution $\sum_{u,v} q(u,v) = 1$ translates to $\langle \alpha_{e} | \identity \otimes \Phi_{e} | \beta_{e} \rangle = 1$. The quantum value $\omega_{q}(\textsl{g}^{l})$ of the linear game can therefore be bounded as
\begin{eqnarray}
\omega_{q}(\textsl{g}^{l}) &=& \frac{1}{|G|} \sum_{x \in G} \langle \alpha_{x} | \identity \otimes \Phi_{x} | \beta_{x} \rangle \nonumber \\
& \leq & \frac{1}{|G|} \left[ 1 + \sqrt{\vert Q_A \vert \vert Q_B \vert} \sum_{x \in G\setminus \{e\}} \Vert \Phi_{x} \Vert \right],
\end{eqnarray}
where $\Vert \Phi_{x} \Vert$ denotes the norm of the game matrices $\Phi_{x}$. 
For games where the winning constraint only depends upon the $\textsc{xor}$ of the outcomes, i.e. $V(a,b|u,v) = 1$ iff $a \oplus_d b = f(u,v)$ for $u \in [m_A], v \in [m_B]$ and $f(u,v) \in \{0,\dots, d-1\}$, the above reduces to 
\begin{eqnarray}
\label{xor-d-bound}
\omega_{q}(\textsl{g}^{\oplus}) &=& \frac{1}{d} \sum_{k =0}^{d-1} \langle \alpha_{k} | \identity \otimes \Phi_{k} | \beta_{k} \rangle \nonumber \\
& \leq & \frac{1}{d} \left[ 1 + \sqrt{m_A m_B} \sum_{k=1}^{d-1} \Vert \Phi_{k} \Vert \right].
\end{eqnarray}
\end{proof}

\section{Linear games with no quantum advantage: Non-local computation}
\label{sec:app-nlc}
We now consider the generalization of the non-local computation task to \textsc{xor}-d games, namely the computation of the function $g(z_1, \dots, z_n)$ with $z_i \in \{0, \dots, d-1\}$ where $d$ is a prime. In these $NLC_d$ games, Alice and Bob receive $n$ dits $\textbf{x}_n = (x_1, \dots, x_n)$ and $\textbf{y}_n = (y_1, \dots, y_n)$ which obey $x_i \oplus_d y_i = z_i$. Their task is to output dits $a, b$ respectively such that
\begin{equation}
\label{func-NLC}
a \oplus_d b = g(\textbf{x}_{n-1} \oplus_d \textbf{y}_{n-1}) \cdot (x_n\oplus_d y_{n}),
\end{equation}
where $\textbf{x}_{n-1} \oplus_d \textbf{y}_{n-1}$ is the dit-wise \textsc{xor} of the $n-1$ dits, i.e., $\{x_1 \oplus_d y_1, \dots, x_{n-1} \oplus_d y_{n-1}\}$ and $g$ is an arbitrary function from $n-1$ dits to $1$ dit. The inputs are chosen according to
\begin{equation}
\label{probdistr2}
\frac{1}{d^{n+1}} p(\textbf{x}_{n-1} \oplus_d \textbf{y}_{n-1})
\end{equation}
for $p(\textbf{x}_{n-1} \oplus_d \textbf{y}_{n-1})$ being an arbitrary probability distribution. 

\begin{apptheorem}
\label{thm:app-nlc}
The games $NLC_d$ for arbitrary prime $d$ and for input distribution satisfying \eqref{probdistr2} have no quantum advantage, i.e., $\omega_c(NLC_d) = \omega_q(NLC_d)$.
\end{apptheorem}

\begin{proof}
We first consider the case of uniformly chosen inputs. The games $NLC_d$ consider functions of the following form (all arithmetic operations being performed modulo $d$)
\begin{equation}
\label{func-NLC-2}
a \oplus_d b = g(x_1 \oplus_d y_1, \dots, x_{n-1} \oplus_d y_{n-1}) \cdot (x_n\oplus_d y_{n}),
\end{equation}
with $g$ being an arbitrary function. Such a game is therefore composed of ``building-block games" $G(t)$ which are of the form
\begin{equation}
\label{single-game}
G(t):=\DE{a \oplus_d b= t \cdot (x \oplus_d y)},
\end{equation}
with $t \in \{0, \dots, d-1\}$, i.e., $f(x,y) = t \cdot (x \oplus_d y)$. There are $d$ different games $G(t)$, each with single dit input for each party (which we will take to be $x_n$ and $y_n$), and these games all have classical value $\omega_c(G(t))= 1 \; \forall t$. Explicitly,the classical strategy $a = t \cdot x$ and $b = t \cdot y$ wins the game $G(t)$. We can write the corresponding (non-normalized) game matrices $\Phi^{(1)}_{k}(t)$ for games $G(t)$ and they take the form
\begin{equation}
\label{single-game-mat}
\Phi^{(1)}_{k}(t) := \sum_{x,y \in \{0, \dots, d-1\}} \zeta^{k t (x \oplus_d y)} |x \rangle \langle y|,
\end{equation}
with $\zeta = \exp{(2\pi I/d)}$. Here the $(1)$ in the superscript denotes that these matrices correspond to the $NLC_d$ game matrices for $n=1$. Let us analyze some properties of the $\Phi^{(1)}_{k}(t)$. Firstly, we see that ${\Phi^{(1)}_{k}(t)}^{\dagger} \Phi^{(1)}_{k}(t)$ for any $k, t$ is diagonal in the Fourier basis defined by the Fourier vectors $|f_j\rangle$ with
\begin{equation}
|f_{j} \rangle = \left(1, \zeta^j, \zeta^{2j}, \dots, \zeta^{(d-1)j}\right)^T
\end{equation}
with $j \in \{0, \dots, d-1\}$. Moreover, we also see that each ${\Phi^{(1)}_{k}(t)}^{\dagger} \Phi^{(1)}_{k}(t)$ has only one eigenvalue (=$d^2$) different from zero and this corresponds to the eigenvector $|f_{d-k \cdot t}\rangle$. This gives the orthogonality ${\Phi^{(1)}_{k}(t)}^{\dagger} \Phi^{(1)}_{k'}(t') = 0$ for $k \cdot t \neq k' \cdot t'$. Since, we will be concerned with finding the maximum singular vectors corresponding to a fixed $k$, we can encapsulate the above properties by the equation 
\begin{equation}
\label{single-game-prop}
\left[ {\Phi^{(1)}_{k}(t)}^{\dagger} \Phi^{(1)}_{k}(t') \right] |f_{j} \rangle = d^2 \delta_{t, t'} \delta_{j, d- k.t}  |f_{j} \rangle
\end{equation}
We shall use these properties of the $\Phi^{(1)}_k(t)$ as we proceed to analyze the game matrices $\Phi^{(n)}_k$ for the general $n$ dit input $NLC_d$ games themselves.

Consider the games $NLC_d$ for prime $d$ and arbitrary number $n$ of input dits for each party. Denote the total number of inputs for each party by $m=d^n$, and the corresponding game matrices by ${\Phi^{(n)}_k}$. Due to the structure of the function in Eq. (\ref{func-NLC-2}), namely the fact that the games only depend on the dit-wise {\sc xor} 
of the $n$ dits, we see that ${\Phi^{(n)}_k}^{\dagger} \Phi^{(n)}_k$ acquires a block circulant structure
(for $1 \leq i \leq n$ the corresponding matrices ${\Phi^{(i)}_k}^{\dagger} \Phi^{(i)}_k$ for each $k$ are block-wise circulant matrices). For example, a possible (unnormalized) game matrix $\Phi_{ex}$ for $n=2, d=3$ of the form  
\begin{equation}
\scalebox{0.7}{ \begin{tabular}{| l |  l | r  | }
  \hline
    $ \Phi^{(1)}(0)$ & $\Phi^{(1)}(1)$  & $\Phi^{(1)}(2)$   \\ \hline
    $ \Phi^{(1)}(1)$ & $\Phi^{(1)}(2)$  & $\Phi^{(1)}(0)$  \\  \hline
     $\Phi^{(1)}(2)$  & $\Phi^{(1)}(0)$ & $\Phi^{(1)}(1)$  \\ \hline
  \end{tabular}  }
 \end{equation}
with the $\Phi^{(1)}(t)$ defined as in Eq.(\ref{single-game-mat}) would have $\Phi_{ex}^{\dagger} \Phi_{ex}$ equal to 
\begin{equation}
\scalebox{0.7}{ \begin{tabular}{| l  | l | r  | }
  \hline
     $\sum_{i} \Phi^{(1)}(i)^{\dagger} \Phi^{(1)}(i)$ & $\sum_{i} \Phi^{(1)}(i)^{\dagger} \Phi^{(1)}(i+1)$  & $\sum_{i} \Phi^{(1)}(i)^{\dagger} \Phi^{(1)}(i+2)$   \\ \hline
     $\sum_{i} \Phi^{(1)}(i)^{\dagger}\Phi^{(1)}(i+2)$ & $\sum_{i} \Phi^{(1)}(i)^{\dagger} \Phi^{(1)}(i)$  & $\sum_{i} \Phi^{(1)}(i)^{\dagger}\Phi^{(1)}(i+1)$  \\  \hline
     $\sum_{i} \Phi^{(1)}(i)^{\dagger}\Phi^{(1)}(i+1)$  & $\sum_{i} \Phi^{(1)}(i)^{\dagger} \Phi^{(1)}(i+2)$ & $\sum_{i} \Phi^{(1)}(i)^{\dagger} \Phi^{(1)}(i)$ \\ \hline
  \end{tabular}  }
\end{equation}
which is a block-wise circulant matrix. 
In general, the entries of ${\Phi^{(n)}_k}^{\dagger} \Phi^{(n)}_k$ are explicitly given by 
\begin{eqnarray}
\De{{\Phi^{(n)}_k}^{\dagger} \Phi^{(n)}_k}_{\vec{x}_{n-1}, \vec{y}_{n-1}} &&= \nonumber \\ 
\sum_{u_1, \dots, u_{n-1} =0}^{d-1} &&{\Phi^{(1)\,{\dagger}}_{k, g(\textbf{x}_{n-1}\oplus_d \textbf{u}_{n-1})}} \Phi^{(1)}_{k,g(\textbf{u}_{n-1}\oplus_d \textbf{y}_{n-1})}
\end{eqnarray}
where as before $\textbf{x}_{n-1} = (x_1, \dots, x_{n-1})$ and $\textbf{y}_{n-1} = (y_1, \dots, y_{n-1})$ are strings of $n-1$ dits, and we have omitted the normalization factor (of $1/d^{4n}$) for clarity. Due to this block circulant structure, we have that ${\Phi^{(n)}_k}^{\dagger} \Phi^{(n)}_k$ for any $n, k$ is diagonal in the basis formed by the tensor products of the Fourier vectors $\{|f_{i_1}\rangle \otimes \dots |f_{i_{n}} \rangle\}$ with $i_1, \dots, i_n \in \{0, \dots, d-1\}$. 

We now proceed to calculate the maximum eigenvector of ${\Phi^{(n)}_{k}}^{\dagger} \Phi^{(n)}_{k}$ among the basis formed by $\{|f_{i_1}\rangle \otimes \dots |f_{i_{n}} \rangle\}$. To do this, let us consider the case of fixed $i_n$ vary $i_1, \dots i_{n-1}$. 
Using the properties of the game matrices $\Phi^{(1)}_{k}(t)$ encapsulated by Eq. (\ref{single-game-prop}), we see that for any fixed $i_n$, the eigenvalue corresponding to $|f_0 \rangle^{\otimes n-1} \otimes |f_{i_n} \rangle$ cannot be smaller than that corresponding to any other $|f_{i_1} \rangle \otimes \dots |f_{i_n} \rangle$. This is due to the fact that the other eigenvectors contribute only phases $\zeta^j$ to the eigenvalue expression corresponding to $|f_0 \rangle^{\otimes n-1} \otimes |f_{i_n} \rangle$ and the properties stated above. 
It therefore follows that the maximum eigenvector is among the $|f_0\rangle^{\otimes n-1} \otimes |f_{i_n} \rangle$.

Let us compute the eigenvalues corresponding to $|f_0\rangle^{\otimes n-1} \otimes |f_{i_n} \rangle$ for $i_n \in \{0, \dots, d-1\}$. To do this, fix an input string $\textbf{x}_{n-1}$ (to say $(0,\dots,0)$) and vary over $\textbf{y}_{n-1}$, in other words we consider the first row block of
$\Phi^{(n)}_k$ corresponding to the game blocks $\Phi^{(1)}_{k,g(\textbf{0}_{n-1} \oplus_d \textbf{y}_{n-1})}$ of size $d \times d$. Denote by $\lambda^{\textbf{x}_{n-1}}(i_n, k)$ the number of times the 
game $G{(d-k \cdot i_n)}$ appears for this  choice of $\textbf{x}_{n-1}$ in matrix $\Phi^{(n)}_k$. Due to the symmetry of the game constraint, 
$\lambda^{\textbf{x}_{n-1}}(i_n, k)$ is independent of the choice of row $\textbf{x}_{n-1}$ so we may drop the superscipt. Moreover, since $\Phi^{(n)}_k$ is
a symmetric matrix, we also have $\lambda^{\textbf{x}_{n-1}}(i_n, k) = \lambda^{\textbf{y}_{n-1}}(i_n, k)$ for an analogously 
defined $\lambda^{\textbf{y}_{n-1}}(i_n, k)$. Let us define $\Lambda(k) := \max_{i_n} \lambda(i_n, k)$
and let $\mu \in \{0, \dots, d-1\}$ denote 
the value of $i_n$ for which the maximum of $\lambda(i_n, k)$ is achieved. 
Again using Eq. (\ref{single-game-prop}), we have that 
\begin{equation}
\left[{\Phi^{(n)}_k}^{\dagger} \Phi^{(n)}_k \right] |f_0 \rangle^{\otimes n-1} \otimes |f_{i_n} \rangle = d^2 \lambda^2(i_n, k) |f_0 \rangle^{\otimes n-1} \otimes |f_{i_n} \rangle.
\end{equation}
We therefore obtain that $\| \Phi^{(n)}_k \| = d \Lambda(k)$. 

For prime $d$, we see that $\Lambda(k) = \Lambda$, constant and independent of $k$. This follows from the fact that the number of 
generators of the additive group $\mathbb{Z}_d$ for prime $d$ is simply equal to $d-1$ (all numbers less than prime $d$ are relatively prime to it). Therefore, for prime $d$, we obtain the following bound
on the quantum value in the uniform case
\begin{equation}
\label{uni-q-bound}
\omega_q(NLC_d)  \leq  \frac{1}{d}\left(1 + \frac{(d-1) \Lambda}{d^{n-1}} \right).
\end{equation}
We now consider the classical deterministic strategy where Alice outputs $a = \mu x_n$ independently of her input $\textbf{x}_{n-1}$ and Bob outputs $b = \mu y_n$ 
independently of his input $\textbf{y}_{n-1}$. Note that for the $d \times d$ blocks described by $G(\mu)$ all the $d^2$ constraints will be satisfied. On the other hand, for the blocks described by $G(t)$ for $t \neq \mu)$, only $d$ constraints are satisfied with the use of this strategy.
The score achieved by this strategy is therefore given by
\begin{equation}
\omega_c(NLC_d) = \frac{d^{n-1}}{d^{2n}}\De{\Lambda d^2 + (d^{n-1} - \Lambda) d},
\end{equation}
which equals the upper bound on the quantum value in Eq. (\ref{uni-q-bound}); this completes the proof for uniformly chosen inputs.

Having solved the problem for uniformly distributed inputs, we can generalize to the case of probability distributions
\begin{equation}
\frac{1}{d^{n+1}} p(\textbf{x}_{n-1}\oplus_d \textbf{y}_{n-1})
\end{equation}
For this input distribution, the matrix ${\Phi}_{k}^{(n)}$ is still composed of the elementary games $\Phi^{(1)}_{k}(t)$ that can be classically saturated. The difference is that a weight $ p(\textbf{x}_{n-1} \oplus_d \textbf{y}_{n-1})/d^{n+1}$
is now attributed to each element of the $d \times d$ block
\begin{equation}
[{\Phi}^{(n)}_k]_{\textbf{x}_{n-1}, \textbf{y}_{n-1}} = \frac{1}{d^{n+1}} p(\textbf{x}_{n-1} \oplus_d \textbf{y}_{n-1}) \Phi^{(1)}_{k,g(\textbf{x}_{n-1}\oplus_d \textbf{y}_{n-1})}.
\end{equation}
This preserves the block-wise circulant structure of ${\Phi^{(n)}_{k}}^{\dagger} \Phi^{(n)}_{k}$ ensuring that these matrices are still diagonal in the basis formed by the tensor products of Fourier vectors. 
As in the case of uniformly distributed inputs, the properties of $\Phi^{(1)}_{k}(t)$ in Eq. (\ref{single-game-prop}) imply that the maximum eigenvector corresponds to one choice of $i_n \in \{0, \dots, d-1\}$ within the $|f_0\rangle^{\otimes n-1} \otimes |f_{i_n} \rangle$. 

To compute the eigenvalues corresponding to $|v_0\rangle^{\otimes n-1} \otimes |v_{i_n} \rangle$, we have to take into account the number of times a game $G{(d-k \cdot i_n)}$ appear in a given row block as well as the respective weights.
Denote by $\tilde{\lambda}(i_n, k)$ the weighted sum of the times the game $G{(d-k \cdot i_n)}$ appears in a row block, i.e.,
\begin{align}
\tilde{\lambda}(i_n, k)=\sum_{\stackrel{\textbf{y}_{n-1} \text{ s.t.}}{ g(\vec{0}_{n-1}\oplus_d \textbf{y}_{n-1})=i_n}}  \frac{1}{d^2} p(\textbf{0}_{n-1} \oplus_d \textbf{y}_{n-1})
\end{align}
As before, let us define $\tilde{\Lambda}(k) := \max_{i_n} \tilde{\lambda}(i_n, k)$ and let $\mu$ denote the $i_n$ for which the maximum is reached.
For the weighted matrix we have
\begin{equation}
\left[{\tilde{\Phi}^{(n)\dagger}_k} \tilde{\Phi}^{(n)}_k \right] |f_0 \rangle^{\otimes n-1} \otimes |f_{i_n} \rangle = d^2 \tilde{\lambda}(i_n, k)^2 |f_0 \rangle^{\otimes n-1} \otimes |f_{i_n} \rangle.
\end{equation}
We therefore obtain that $\| \tilde{\Phi}^{(n)}_k \| = d \tilde{\Lambda}(k)$. 

Again, for prime $d$, the maximum of this sum is independent of  $k$. Therefore, for prime $d$, we obtain the following bound on the quantum value for a general $NLC_d$ game
\begin{equation}
\label{nlc-q-bound}
\omega_q(NLC_d)  \leq  \frac{1}{d}\De{1 + d^{n+1}(d-1) \tilde{\Lambda}}.
\end{equation}
Consider the classical deterministic strategy where Alice outputs $a = \mu x_n$ independently of $\textbf{x}_{n-1}$ and Bob outputs $b = \mu y_n$ 
independently of $\textbf{y}_{n-1}$. For the $d \times d$ blocks described by $G(\mu)$ all the $d^2$ constraints will be satisfied. On the other hand, for the blocks described by $G(t \neq \mu)$, only $d$ constraints are satisfied with the use of this strategy.
The score achieved by this strategy is therefore given by
\begin{equation}
\omega_c(NLC_d) = d^{n-1} \De{\tilde{\Lambda} d^2 + \de{\frac{1}{d^{n+1}} - \tilde{\Lambda}}d },
\end{equation}
which equals the upper bound on the quantum value in Eq.(\ref{nlc-q-bound}); this completes the proof that quantum strategies cannot outperform classical ones in the $NLC_d$ game.
\end{proof}


\end{document}